\documentclass{article}
\usepackage{latexsym, amssymb,amsmath, amsbsy, amsopn, amstext, amsthm}
\usepackage{graphicx}
\usepackage{hyperref}
\usepackage{float}
\usepackage[textwidth=16cm,textheight=21.5cm]{geometry}

\usepackage{tikz}

 \DeclareMathOperator{\Col}{Col}

 \DeclareMathOperator{\lcm}{lcm}

\def\ra{\rightarrow}

\def\lra{\leftrightarrow}

\def\a{\alpha}

\def\d{\delta}

\def\D{\Delta}

\def\0{{\bf 0}}
\def\J{J}

\newcommand{\R}{{\mathbb R}}

\newtheorem{thm}{Theorem}[section]
\newtheorem{lem}[thm]{Lemma}

\newtheorem{dfn}[thm]{Definition}
\newtheorem{prp}[thm]{Proposition}
\newtheorem{exa}[thm]{Example}
\newtheorem{rem}[thm]{Remark}
\newtheorem{alg}[thm]{Algorithm}
\begin{document}

\begin{center}

{\Large\bf Invariant Subspace Approach to Boolean (Control) Networks
\footnote{ Supported partly by NNSF 62073315 of China, and China Postdoctoral Science Foundation 2020TQ0184.} }
\end{center}

\vskip 2mm

\begin{center}
{\Large  Daizhan Cheng$~^*$\dag, Lijun Zhang\ddag, Dongyao Bi\ddag}

\vskip 2mm

$~^*$ Center of STP Theory and Applications, LiaoCheng University, LiaoCheng 252000, P.R.China\\
\dag Institute of  Systems Science,
Chinese Academy of Sciences, Beijing 100190, P.R.China\\
E-mail: dcheng@iss.ac.cn\\

\ddag Northwestern Polytechnical University\\

\vskip 2mm

\end{center}

\vskip \baselineskip

\underline{\bf Abstract}: A logical function can be used to characterizing a property of a state of Boolean network (BN), which is considered as an aggregation of states.  To illustrate the dynamics of a set of logical functions, which characterize our concerned properties of a BN, the invariant subspace containing the set of logical functions is proposed, and its properties are investigated. Then the invariant subspace of Boolean control network (BCN) is also proposed. The dynamics of invariant subspace of BCN is also invariant. Finally, using outputs as the set of logical functions, the minimum realization of BCN is proposed, which provides a possible solution to overcome the computational complexity of large scale BNs/BCNs.

\vskip 5mm

\underline{\bf Keywords}:
Boolean (control) network, logical function, invariant subspace, minimum realization, semi-tensor product.

\vskip 5mm

\section{Introduction}

The BN was firstly proposed  by Kauffman in 1969 \cite{kau69}. It has been proved to be a very efficient way for modeling and analyzing genetic regulatory network. Recently,  motivated by the semi-tensor product (STP) of matrices, the investigation of BN and BCN becomes a heat research direction in control community. Nowadays, the STP approach becomes the mainstream in studying BNs and BCNs. We refer to some survey papers for its current development in theory and applications \cite{for16}, \cite{lu17}, \cite{muh16}, \cite{li18}.

The major obstacle in applications of STP approach to BNs and BCNs is the computational complexity. It is well known that BN structure analysis and BCN control design and many related problems are NP hard problem \cite{zha05}, while BNs from gene regulatory networks are usually of large scale. For a network with $n$ nodes, the state space of BN or BCN in STP model  is of $2^n$ states. Hence, in general, the STP approach can only handle $n<20$ cases or so.

A proper tool in dealing with large-scale BN (BCN) is aggregation \cite{zha13,zha15,pan20}. To the authors' best knowledge, the aggregation proposed so far is based on the structure of networks. This method has some weaknesses. First, it requires the knowledge on the structure of networks. It is not an easy job to get the structure of a large scale network. Second, such an aggregation does not represent certain properties of the nodes. Sometimes, classifying nodes according to their various properties is more important then their positions.

Support vector machine approach is a very powerful tool in aggregation, where it is called pattern recognition \cite{bur98,hay09}.
In support vector machine a hyperplane $w^Tx+b$, which separates points into two groups: $\{x\;|\;w^Tx+b>0\}$ and $\{x\;|\;w^Tx+b<0$, (refer to Fig. \ref{Fig.1.1}).

\vskip 2mm

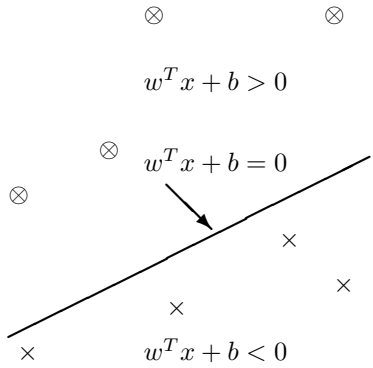
\begin{figure}
\centering
\setlength{\unitlength}{6mm}
\begin{picture}(10,10)\thicklines
\put(1,1){\line(2,1){8}}
\put(1,4){$\otimes$}
\put(1.2,0.5){$\times$}
\put(8,8){$\otimes$}
\put(8.2,2){$\times$}
\put(3,5){$\otimes$}
\put(4.5,1.5){$\times$}
\put(4,8){$\otimes$}
\put(7,3){$\times$}
\put(4,4.7){$w^Tx+b=0$}
\put(4,6.5){$w^Tx+b>0$}
\put(4,0.5){$w^Tx+b<0$}
\put(4.5,4.4){\vector(1,-1){1}}

\end{picture}
\caption{Hyperplane For Point Separation}\label{Fig.1.1}
\end{figure}

\vskip 2mm

This paper uses the idea of support vector machine to aggregation of nodes in a large-scale Boolean network.
A logical function  $g(x)$ is considered as a support vector, which classifies nodes into two groups: $\{x\;|\;g(x)=1\}$ and $\{x\;|\;g(x)=0$, (refer to Fig. \ref{Fig.1.2}).

\vskip 2mm

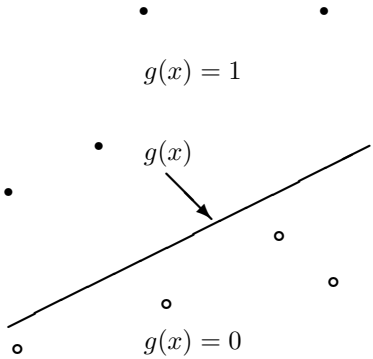
\begin{figure}
\centering
\setlength{\unitlength}{6mm}
\begin{picture}(10,10)\thicklines
\put(1,1){\line(2,1){8}}
\put(1,4){\circle*{0.2}}
\put(1.2,0.5){\circle{0.2}}
\put(8,8){\circle*{0.2}}
\put(8.2,2){\circle{0.2}}
\put(3,5){\circle*{0.2}}
\put(4.5,1.5){\circle{0.2}}
\put(4,8){\circle*{0.2}}
\put(7,3){\circle{0.2}}
\put(4,4.7){$g(x)$}
\put(4,6.5){$g(x)=1$}
\put(4,0.5){$g(x)=0$}
\put(4.5,4.4){\vector(1,-1){1}}

\end{picture}
\caption{Logical Function For State Separation}\label{Fig.1.2}
\end{figure}

\vskip 2mm

Several logical functions, which form a set of support vectors for various properties, become a subspace. Using this subspace, we may construct a logical dynamic system, which describes the dynamics of aggregated classis. Since this dynamic system might be much smaller than the original one, the computational complexity could be reduced a lot.

Roughly speaking, the idea for the approach in this paper is as follows: First, some logical functions are chosen to characterize some properties of a BN/BCN, concerned by us. Then the smallest subspaces containing the set of logical functions, which is invariant under the dynamic evolution. The dynamic equation for the subspace is revealed, which completely described the evolution of the concerned logical variables, which correspond the set of logical functions. Finally, the outputs of a BCN are considered as the set of concerned logical functions, which lead to a minimum realization of the original BCN.

The rest of this paper is follows: Section 2 presents some preliminaries as follows: (i) STP of matrices, which is the fundamental tool for our approach; (ii) Matrix expression of BN and BCN, which is called the algebraic state space representation (ASSR). Section 3 presents the separating subspace approach for BNs. The separating logical functions and the invariant subspace containing the set of logical functions is constructed, and its properties are investigated. Finally, the dynamic equation is obtained for the invariant subspace. The invariant subspace and its dynamic equation of BCN are considered in Section 4. Section 5 considers the minimum realization of a BCN. Their dynamic equations are also revealed. Section 6 is a brief conclusion.

\section{Preliminaries}

\subsection{STP of Matrices}

\begin{dfn} \label{d2.1.1} \cite{che11, che12}:
Let $M\in {\cal M}_{m\times n}$, $N\in {\cal M}_{p\times q}$, and $t=\lcm\{n,p\}$ be the least common multiple of $n$ and $p$.
The semi-tensor product (STP) of $M$ and $N$, denoted by $M\ltimes N$, is defined as
\begin{align}\label{2.1.1}
\left(M\otimes I_{t/n}\right)\left(N\otimes I_{t/p}\right)\in {\cal M}_{mt/n\times qt/p},
\end{align}
where $\otimes$ is the Kronecker product.
\end{dfn}

Note that when $n=p$, $M\ltimes N=MN$. That is, the semi-tensor product is a generalization of conventional matrix product. Moreover, it keeps all the properties of conventional matrix product available \cite{che12}. Hence we can omit the symbol $\ltimes$. Throughout this paper the matrix product is assumed to be STP, and the symbol $\ltimes$ is mostly omitted.

The following are some basic properties:

\begin{prp}\label{p2.1.2}
\begin{enumerate}
\item (Associative Law)
\begin{align}\label{2.1.2}
(F\ltimes G)\ltimes H = F\ltimes (G\ltimes H).
\end{align}
\item (Distributive Law)
\begin{align}\label{2.1.3}
\begin{cases}
F\ltimes (aG\pm bH)=aF\ltimes G\pm bF\ltimes H,\\
(aF\pm bG)\ltimes H=a F\ltimes H \pm bG\ltimes H,\quad a,b\in \R.
\end{cases}
\end{align}
\end{enumerate}
\end{prp}

\begin{prp}\label{p2.1.3}
\begin{enumerate}
\item Let $X\in \R^m$, $Y\in \R^n$ be two columns. Then
\begin{align}\label{2.1.4}
X\ltimes Y=X\otimes Y.
\end{align}
\item Let $\omega \in \R^m$, $\sigma\in \R^n$ be two rows. Then
\begin{align}\label{2.1.5}
\omega\ltimes \sigma =\sigma \otimes \omega.
\end{align}
\end{enumerate}
\end{prp}

About the transpose, we have
\begin{prp}\label{p2.1.5}
\begin{align}
\label{2.1.6} (A\ltimes B)^\mathrm{T}=B^\mathrm{T}\ltimes A^\mathrm{T}.
\end{align}
\end{prp}

About the inverse, we have

\begin{prp}\label{p2.1.6}
Assume $A$ and $B$ are invertible, then
\begin{align} \label{2.1.7}
(A\ltimes B)^{-1}=B^{-1}\ltimes A^{-1}.
\end{align}
\end{prp}

The following property is for STP only.

\begin{prp}\label{p2.1.7} Let $X\in \R^m$ be a column and $M$ a matrix. Then
\begin{align}\label{2.1.8}
X\ltimes M=\left(I_m\otimes M\right)X.
\end{align}
\end{prp}

\begin{dfn}\label{d2.1.8} \cite{che11} A matrix $W_{[m,n]}\in {\cal M}_{mn\times mn}$, defined by
\begin{align}\label{2.1.9}
W_{[m,n]}:=\left[I_n\otimes \d_{m}^1, I_n\otimes \d_{m}^2,\cdots,I_n\otimes \d_{m}^m, \right]
\end{align}
is called the $(m,n)$-th dimensional swap matrix, where $\d_m^i$ is the $i$-th column of $I_m$.
\end{dfn}

The basic function of the swap matrix is to ``swap" two vectors. That is,

\begin{prp}\label{p2.1.10} Let $X\in \R^m$ and $Y\in \R^n$ be two columns. Then
\begin{align}\label{2.1.11}
W_{[m,n]}\ltimes X\ltimes Y=Y\ltimes X.
\end{align}
\end{prp}

\begin{dfn}
Let $A\in \mathcal{M}_{p\times n}$ and $B\in \mathcal{M}_{q\times n}$. Then the Khatri-Rao Product of $A$ and $B$ is
  \begin{align}
  \begin{array}{ccl}
  A*B=[\Col_1(A)\ltimes \Col_1(B),\cdots,\Col_n(A)\ltimes \Col_n(B)]
  \in \mathcal{M}_{pq\times n}.
  \end{array}
  \end{align}
\end{dfn}

\subsection{Matrix Expression of BN}

\begin{dfn}\label{d2.2.1} A BN is described by
\begin{align}\label{2.2.1}
\begin{cases}
x_1(t+1)=f_1(x_1(t),\cdots,x_n(t)),\\
x_2(t+1)=f_2(x_1(t),\cdots,x_n(t)),\\
\vdots\\
x_n(t+1)=f_n(x_1(t),\cdots,x_n(t)),\\
\end{cases}
\end{align}
where $x_i(t)\in {\cal D}=\{0,1\}$, $f_i:{\cal D}^n\ra {\cal D}$, $i=1,2,\cdots,n$ are logical functions.
\end{dfn}

Using vector form expression: $1\sim \d_2^1=(1,0)^T$, $0\sim \d_2^2=(0,1)^T$. Then $x(t)$ can be expressed as $x(t)\in \D_2$, where $\D_k$ is the set of columns of $I_k$.

A  matrix $M\in {\cal M}_{p\times q}$ is called a logical matrix, if $\Col(M)\subset \D_p$. The set of $p\times q$ dimensional logical matrices is denoted by ${\cal L}_{p\times q}$.

Then a BN has its matrix form, called the ASSR of BN, as follows:

\begin{prp}\label{p2.2.2}
\begin{itemize}
\item[(i)] For a logical function $f:{\cal D}^n\ra {\cal D}$, there exists a unique logical matrix $M_f\in {\cal L}_{2\times 2^n}$ such that in vector form
\begin{align}\label{2.2.2}
f(x_1,x_2,\cdots,x_n)=M_f\ltimes_{i=1}^nx_i.
\end{align}
\item[(ii)] Let $M_i$ be the structure matrix of $f_i$, $i=1,2,\cdots,n$. Then in vector form (\ref{2.2.1}) can be expressed into its componentwise ASSR as
\begin{align}\label{2.2.3}
\begin{cases}
x_1(t+1)=M_1\ltimes_{i=1}^n x_i(t),\\
x_2(t+1)=M_2\ltimes_{i=1}^nx_i(t),\\
\vdots,\\
x_n(t+1)=M_n\ltimes_{i=1}^n x_i(t).\\
\end{cases}
\end{align}
\item[(iii)] Setting $x(t)=\ltimes_{i=1}^nx_i(t)$, (\ref{2.2.3}) can further be expressed into its ASSR as
\begin{align}\label{2.2.4}
x(t+1)=Mx(t),
\end{align}
where
$$
M=M_1*M_2*\cdots*M_n\in {\cal L}_{2^n\times 2^n}
$$
is called the structure matrix of BN (\ref{2.2.1}).
\end{itemize}
\end{prp}

Similarly, the BCN is described as follows:

\begin{align}\label{2.2.5}
\begin{cases}
x_1(t+1)=f_1(x_1(t),\cdots,x_n(t);u_1(t),\cdots,u_m(t)),\\
x_2(t+1)=f_2(x_1(t),\cdots,x_n(t);u_1(t),\cdots,u_m(t)),\\
\vdots,\\
x_n(t+1)=f_n(x_1(t),\cdots,x_n(t);u_1(t),\cdots,u_m(t)),\\
\end{cases}
\end{align}
where $u_j(t)\in {\cal D}$, $j=1,\cdots,m$ are controls.

We also have similar algebraic expressions for BCN.

\begin{prp}\label{p2.2.3} Consider BCN (\ref{2.2.5}).
\begin{itemize}
\item[(i)] Its componentwise ASSR is
\begin{align}\label{2.2.6}
\begin{cases}
x_1(t+1)=L_1u(t)x(t),\\
x_2(t+1)=L_2u(t)x(t),\\
\vdots,\\
x_n(t+1)=L_nu(t)x(t),\\
\end{cases}
\end{align}
where $u(t)=\ltimes_{j=1}^mu_j(t)$, $L_i\in {\cal L}_{2\times 2^{m+n}}$ is the structure matrix of $f_i$, $i=1,\cdots,n$.

\item[(ii)]
Its ASSR is
\begin{align}\label{2.2.7}
x(t+1)=Lu(t)x(t),
\end{align}
where $L=L_1*L_2*\cdots*L_n\in {\cal L}_{2^n\times 2^{m+n}}$ is  is the structure matrix of BCN (\ref{2.2.5}).
\end{itemize}
\end{prp}

\begin{dfn}\label{d2.2.4} \cite{che10} Consider BN
(\ref{2.2.1}) or BCN (\ref{2.2.2}).
\begin{itemize}
\item[(i)] Their state space, denoted by ${\cal X}$, is defined as the set of all logical functions of  ~$x_1, x_2, \cdots, x_n$, denoted by
${\cal F}_{\ell}\{x_1, x_2, \cdots, x_n\}$. That is,
\begin{align} \label{2.2.8}
{\cal X}={\cal F}_{\ell}\{x_1, x_2, \cdots, x_n\}.
\end{align}

\item[(ii)] Let $z_1, z_{2}, \cdots, z_r\in {\cal X}$. Then the subspace generated by ~$z_1, z_2, \cdots, z_r$, is defined by
\begin{align} \label{2.2.9}
{\cal Z}={\cal F}_{\ell}\{z_1, z_2, \cdots, z_r\}.
\end{align}
\end{itemize}
\end{dfn}

Consider (\ref{2.2.9}). Since $z_i\in {\cal X}$, there is a structure matrix of $z_i$, denoted by $G_i$, such that in vector form we have
\begin{align} \label{2.2.10}
z_i=G_ix,\quad i=1,\cdots, r.
\end{align}
where $x=\ltimes_{i=1}^nx_i$, $G_i\in {\cal L}_{2\times 2^n}$.
Denote $z=\ltimes_{i=1}^rz_i$, then we have
\begin{align} \label{2.2.11}
z=Tx,
\end{align}
where $T=G_1*G_2*\cdots*G_r\in {\cal L}_{2^r\times 2^n}$.

\begin{dfn}\label{d2.2.5} Support ${\cal Z}={\cal F}_{\ell}\{z_1,z_2,\cdots,z_n\}$ has its algebraic expression (\ref{2.2.11}) with non-singular $T$, then ${\cal X}\ra {\cal Z}$ is called a coordinate change.
\end{dfn}

\begin{rem}\label{r2.2.6} Since $T$ is a logical matrix, if $T$ is non-singular, then it is a permutation matrix. Hence $T^{-1}=T^T$.
Now if $f\in {\cal X}$, which can be expressed via its structure matrix $M_f$ as
$$
f(x)=M_fx,
$$
then it can also be expressed as
$$
f(x)=\tilde{f}(z)=M_fT^Tz.
$$
\end{rem}

\section{Separating Subspace Approach to BN}

\subsection{Separating Function and Invariant Subspace}

Note that a logical function $f(x_1,x_2,\cdots,x_n)$ can be considered as an index function of a subset of node set $N:=\{x_1,x_2,\cdots,x_n\}$. Given an $S\subset N$, then its index function, denoted by $f_S$, can be defined as follows:
\begin{align}\label{3.1.0}
f_S(x):=\begin{cases}
1,\quad x\in S,\\
0,\quad \mbox{Otherwise}.
\end{cases}
\end{align}
Let $\pi:2^N\ra {\cal F}_{\ell}\{x_1,x_2,\cdots,x_n\}$ determined by $\pi(S)=f_s$, which is defined by (\ref{3.1.0}). Then it is obvious that  $\pi$ is bijective. Based on this observation we can define separating logical function.

For a large-scale BN, if there are $n$ nodes, its number of states is $2^n$. Say, $n=32$, then the states are $4.295E+9$. So in its ASSR, the transition matrix, which has $2^n\times 2^n$ dimension, is not  practically computable. In fact, we may not be interested in its detailed state evolution. We are only interested in some particular properties of the BN. Using the idea of separating logical function approach, the set of separating logical functions of BN is proposed as follows: Assume we are interested in a property, say $p$, for a BN. We may define a logical function $z_p$ as follows:
$$
z_p(x)=
\begin{cases}
1,\quad \mbox{p is true at }x,\\
0,\quad \mbox{p is false at }x,
\end{cases}
$$
$z_p\in {\cal X}$. Such $z_p$ is called a separating function, which classifies   all states into two groups according to property $p$.

In general, we are interested in a set of $z_i(x)$, $i=1,2,\cdots,r$, where $r<<n$. We then can aggregate $x$ into $2^r$ groups as
\begin{align} \label{3.1.1}
x^k:=\left\{x\|z(x)=\d_{2^r}^k\right\},\quad k=1,\cdots,2^r.
\end{align}

\begin{dfn}\label{d3.1.1} Let $z_i$, $i=1,\cdots,r$ be a set of separating logical functions.
Then  $\{z_i\;|\;i=1,2,\cdots,r\}$ are called aggregating variables, and
$$
{\cal Z}:={\cal F}_{\ell}\{z_1,z_2,\cdots,z_r\}
$$
is called the $\{z_i\;|\;i=1,2,\cdots,r\}$ aggregated subspace.
\end{dfn}

Support we are only interested the dynamics about ${\cal Z}$, which might be much more smaller than the original BN.

To find the dynamics of ${\cal Z}$, we need some new concepts.

\begin{dfn}\label{d3.1.2} Given BN (\ref{2.2.1}).
\begin{itemize}
\item[(i)]
${\cal Z}^1={\cal F}_{\ell}\{z^1\}={\cal F}_{\ell}\{z^1_1, z^1_2, \cdots, z^1_r\}$
 is called a regular subspace, if there exist  ~$z^2=(z^2_1, z^2_2, \cdots, z^2_{n-r})$, such that $z=(z^1_1,\cdots,z^1_r,z^2_1,\cdots, z^2_{n-r})$ is another coordinate frame. That is,
 ${\cal X}={\cal Z}={\cal F}_{\ell}\{z^1, z^2\}$.
\item[(ii)] Assume  ${\cal Z}^1$ is a regular subspace and $z=(z^1,z^2)$ is a new coordinate frame.
 Moreover, under $z$, (\ref{2.2.1}) can be expressed as
\begin{align}\label{3.1.2}
\begin{cases}
z^1(t+1)=\tilde{F}^1(z^1(t)), & z^1(t)\in {\cal Z}^1\\
z^2(t+1)=\tilde{F}^2(z(t)),  & z^2(t)\in {\cal Z}^2, z(t)\in {\cal Z},
\end{cases}
\end{align}
then ${\cal Z}^1$ is called an $M$-invariant subspace.
\end{itemize}
\end{dfn}

Recall the ASSR (\ref{2.2.4}) of (\ref{2.2.1}). We have the following result:

\begin{thm}\label{t3.1.3}
Consider BN ~(\ref{2.2.1}) with its ASSR (\ref{2.2.4}).
Suppose ${\cal Z}^1={\cal F}_{\ell}\{z^1_1, z^1_2, \cdots, z^1_r\}$ is a regular space with its ASSR as
\begin{align}\label{3.1.3}
z^1=Qx,
\end{align}
where ~$z^1=\ltimes_{i=1}^rz^1_i$, $Q\in {\cal L}_{2^r\times 2^n}$.  Then,
${\cal
Z}^1$ is an $M$ invariant subspace of ~(\ref{2.2.1}), if and only if, there exists ~$H\in {\cal L}_{2^r\times 2^r}$ such that
\begin{align}\label{3.1.4}
QM=HQ.
\end{align}
\end{thm}

\begin{proof} (sufficiency) Since ~${\cal Z}^1$ is a regular subspace, there exists ~$z^2=(z^2_{1},z^2_{2},\cdots, z^2_{n-r})$, such that ~$z=(z^1,z^2)$ is a new coordinate frame. Hence,
\begin{align}\label{3.1.5}
z^1(t+1)=Qx(t+1)=QMx(t)=HQx(t)=Hz^1(t).
\end{align}
(\ref{3.1.5}) shows that under coordinates ~$z$ the BN has the form of ~(\ref{3.1.2}).

(necessity) Assume under coordinate frame $z$  NB (\ref{2.2.1}) has the form of (\ref{3.1.2}).
Moreover, assume the structure matrix of ~$\tilde{F}_1$ is ~$\tilde{M}_1\in {\cal L}_{2^k\times 2^k}$. Then
$$
z^1(t+1)=\tilde{M}_1z^1(t)=\tilde{M}_1Qx(t).
$$
On the other hand,
$$
z^1(t+1)=Qx(t+1)=QMx(t).
$$
Since ~$x(t)$ is arbitrary, we have
$$
QM=\tilde{M}_1Q.
$$
Set ~$H=\tilde{M}_1$,  ~(\ref{3.1.5}) follows.
\end{proof}

\begin{rem}\label{r3.1.4}
According to ~\ref{t3.1.3}, to verify whether a regular subspace is an invariant subspace we have to check whether equation (\ref{3.1.5}) has solution ~$H$. Since ~${\cal Z}^1$ is a regular subspace, its structure matrix ~$Q$ should have full row rank. Hence, if $H$ is the solution, then $H=H^*$, where
\begin{align}\label{3.1.6}
H^*:=QMQ^{\mathrm{T}}(QQ^{\mathrm{T}})^{-1}.
\end{align}
Hence the see whether (\ref{3.1.5}) has solution ~$H$ we have only to verify if $H^*$ is logical matrix and it satisfies (\ref{3.1.5}).
\end{rem}

We give an example.

\begin{exa}\label{e3.1.5}
Consider the following BN:
\begin{align}\label{3.1.7}
\begin{cases}
x_{1}(t+1)=(x_{1}(t)\wedge x_{2}(t)\wedge \neg x_{4}(t))\vee(\neg x_{1}(t)\wedge x_{2}(t))\\
x_{2}(t+1)=x_{2}(t)\vee(x_{3}(t)\lra x_{4}(t))\\
x_{3}(t+1)=(x_{1}(t)\wedge\neg x_{4}(t))\vee(\neg x_{1}(t)\wedge x_{2}(t))\vee(\neg x_{1}(t)\wedge \neg x_{2}(t)\wedge x_{4}(t))\\
x_{4}(t+1)=x_{1}(t)\wedge\neg x_{2}(t)\wedge x_{4}(t).
\end{cases}
\end{align}
Its ASSR is calculated as
\begin{align}\label{3.1.8}
x(t+1)=Mx(t),
\end{align}
where
$$
M=\delta_{16}[11, 1, 11, 1, 11, 13, 15, 9, 1, 2, 1, 2, 9, 15, 13, 11].
$$
Suppose ${\cal Z}={\cal F}_{\ell}\{z_{1}, z_{2}, z_{3}\}$, where
\begin{align}\label{3.1.9}
\begin{cases}
z_{1}=x_{1}\bar{\vee}x_{4}\\
z_{2}=\neg x_2\\
z_{3}=x_{3}\lra\neg x_{4}.
\end{cases}
\end{align}

Denote ~$x=\ltimes_{i=1}^{4}x_{i}$, $z=\ltimes_{i=1}^{3}z_{i}$, then
$$
z=Qx,
$$
where $Q$ can be calculated as
$$
Q=\d_{8}[8, 3, 7, 4, 6, 1, 5, 2, 4, 7, 3, 8, 2, 5, 1, 6].
$$
Using ~(\ref{3.1.6}), we have
$$
H^*=\d_{8}[2,4,8,8,1,3,3,3].
$$
It is ready to verify ~(\ref{3.1.5}). Hence ~${\cal Z}$ is an invariant subspace of ~(\ref{3.1.7}).
\end{exa}

\subsection{Union of Invariant Subspaces}

Assume ${\cal V}_i$, $i=1,2$ are two $M$ invariant subspaces, where
\begin{align}\label{3.3.1}
\begin{array}{l}
{\cal V}_1={\cal F}_{\ell}\{z^1_1,\cdots,z^1_p\},\\
{\cal V}_2={\cal F}_{\ell}\{z^2_1,\cdots,z^2_q\}.\\
\end{array}
\end{align}
  Then we have
\begin{align}\label{3.3.2}
{\cal V}_i=G_ix,\quad i=1,2,
\end{align}
where $x=\ltimes_{i=1}^nx_i$, $G_1\in {\cal L}_{2^p\times 2^n}$, $G_2\in {\cal L}_{2^q\times 2^n}$.

\begin{thm}\label{t3.3.1} Assume ${\cal V}_i$, $i=1,2$ are $M$ invariant subspaces. That is, there exist $H_1\in {\cal L}_{p\times p}$ and
$H_2\in {\cal L}_{q\times q}$, such that
\begin{align}\label{3.3.3}
\begin{array}{l}
G_1M=H_1G_1,\\
G_2M=H_2G_2.
\end{array}
\end{align}
Then
$${\cal V}={\cal V}_1\bigcup {\cal V}_2={\cal F}_{\ell}\{z^1_1,\cdots,z^1_p;z^2_1,\cdots,z^2_q\}
$$
is also $M$-invariant. Moreover, the structure matrix of ${\cal V}$, denoted by
\begin{align}\label{3.3.4}
G=G_1*G_2,
\end{align}
satisfies
\begin{align}\label{3.3.5}
GM=HG,
\end{align}
where
\begin{align}\label{3.3.6}
H=H_1\otimes H_2.
\end{align}
\end{thm}

To prove this theorem,  we need the following lemma, which itself is useful.

\begin{lem}\label{l3.3.2} Let $A\in {\cal M}_{p\times \ell}$,  $B\in {\cal M}_{q\times \ell}$, and $T\in {\cal L}_{\ell\times r}$.
Then
\begin{align}\label{3.3.7}
(A*B)T=(AT)*(BT).
\end{align}
\end{lem}

\begin{proof}
Denote
$$
A=[A^1,A^2,\cdots,A^{\ell}],\quad B=[B^1,B^2,\cdots,B^{\ell}],
$$
where $A^i=\Col_i(A)$ ($B^i=\Col_i(B)$) is the $i$-th column of $A$ ($B$); and
$$
T=\left[\d_{\ell}^{i_1},\d_{\ell}^{i_2},\cdots,\d_{\ell}^{i_m}\right].
$$
Then
$$
\begin{array}{ccl}
(A*B)T&=&\left([A^1,A^2,\cdots,A^{\ell}]*[B^1,B^2,\cdots,B^{\ell}]\right)T\\
~&=&\left[A^1\otimes B^1,A^2\otimes B^2,\cdots,A^{\ell}\otimes B^{\ell}\right]T\\
~&=&\left[A^{i_1}\otimes B^{i_1},A^{i_2}\otimes B^{i_2},\cdots,A^{i_m}\otimes B^{i_m}\right].
\end{array}
$$
$$
\begin{array}{ccl}
(AT)*(BT)&=&\left[A^{i_1},A^{i_2},\cdots,A^{i_m}\right]*\left[B^{i_1},B^{i_2},\cdots,B^{i_m}\right]\\
~&=&\left[A^{i_1}\otimes B^{i_1},A^{i_2}\otimes B^{i_2},\cdots,A^{i_m}\otimes B^{i_m}\right].
\end{array}
$$
(\ref{3.3.7}) follows immediately.
\end{proof}

\begin{proof} (of Theorem \ref{t3.3.1}) It is enough to prove (\ref{3.3.5}) with (\ref{3.3.6}).
Denote $G_1=(G_1^1,\cdots,G_1^{2^n})$, $G_2=(G_2^1,\cdots,G_2^{2^n})$, where $G_1^i=\Col_i(G_1)$,   $G_2^i=\Col_i(G_2)$, $i=1,2,\cdots, 2^n$. Using Lemma \ref{l3.3.2},
$$
\begin{array}{ccl}
GT&=&(G_1*G_2)T=(G_1T)*(G_2T)\\
~&=&(H_1G_1)*(H_2G_2)\\
~&=&\left[(H_1G_1^1)*(H_2G_2^1), (H_1G_1^2)*(H_2G_2^2),\cdots, (H_1G_1^{2^n})*(H_2G_2^{2^n})\right]\\
~&=&\left[(H_1G_1^1)\otimes(H_2G_2^1), (H_1G_1^2)\otimes (H_2G_2^2),\cdots, (H_1G_1^{2^n})\otimes(H_2G_2^{2^n})\right]\\
~&=&\left[(H_1\otimes H_2)(G_1^1\otimes G_2^1), (H_1\otimes H_2)(G_1^2\otimes G_2^2),\cdots, (H_1\otimes H_2)(G_1^{2^n}\otimes G_2^{2^n})\right]\\
~&=&(H_1\otimes H_2)(G_1*G_2)=(H_1\otimes H_2)G.
\end{array}
$$
\end{proof}

\subsection{Dynamics of Aggregated NB}

Assume (\ref{2.2.1}) is a large scale BN, and  $z_i$, $i=1,\cdots,r$ are separating logical functions, which represent our interested properties. Denote by
$$
{\cal Z}={\cal F}_{\ell}\{z_i\;|\; i=1,\cdots,r\}
$$
We first try to find the smallest subspace $\overline{{\cal Z}}$, which contains ${\cal Z}$ and is $M$-invariant.

\begin{alg}\label{a3.2.1}
\begin{itemize}
\item Step 1: Set $z^0=\ltimes_{i=1}^rz_i$, and assume
$$
z^0=G_0x.
$$
Calculate
$$
z^1=\{G_0x\cup G_0Mx\}:=G_1x.
$$
\item Step k: Assume $z^{k-1}=G_{k-1}x$ is known. Then
$$
z^k=\{G_{k-1}x\cup G_{k-1}Mx\}:=G_kx.
$$
\item Final Step: Assume $z^{k^*}=z^{k^*+1}$, then
\begin{align}\label{3.2.1}
\overline{{\cal Z}}:={\cal F}_{\ell}\{z^{k^*}\}.
\end{align}
\end{itemize}
\end{alg}

\begin{rem}\label{r3.2.101}
In Algorithm \ref{a3.2.1} at each step we assume in $z^i$ all the repeated functions have been deleted. Otherwise, $G_i$ maybe unnecessarily large.
\end{rem}

By construction it is clear that the $\overline{{\cal Z}}$ provided by (\ref{3.2.1}) is the smallest subspace,  containing ${\cal Z}$ and is $M$-invariant.

\begin{dfn}\label{d3.2.2} The dynamics of $\overline{{\cal Z}}$ is called the $\{z_i\;|\; i=1,\cdots,r\}$ aggregated BN.
\end{dfn}

Next, we try to find the dynamics of aggregated BN.

Assume $\overline{{\cal Z}}={\cal F}_{\ell}(\bar{z})$ is a regular subspace, then
$$
\bar{z}=\bar{G}x.
$$
Using Theorem \ref{t3.1.3}, we have that
\begin{align}\label{3.2.2}
\begin{array}{ccl}
\bar{z}(t+1)&=&\bar{G}x(t+1)=\bar{G}Mx(t)\\
~&=&H\bar{G}x(t)=H\bar{z}(t).
\end{array}
\end{align}

Summarizing the above arguments, we have the following result.

\begin{thm}\label{t3.2.3} (\ref{3.2.2}) is  the dynamics of aggregated BN.
\end{thm}

\begin{rem}\label{r3.2.4} It is obvious that in Theorem \ref{t3.2.3} the regularity of  $\overline{{\cal Z}}$ has been ignored. From Algorithm \ref{a3.2.1} one sees easily that (\ref{3.1.4})  is enough for (\ref{3.2.2}). In fact, we do not care about if $\overline{{\cal Z}}$ is regular or not. When it is not, we can not get the second part of equation (\ref{3.1.2}), which is not interesting to us.
\end{rem}

In the following an example is given to describe the technique for constructing aggregated BN.
.

\begin{exa}\label{e3.2.5} An opinion dynamic network is depicted in Fig. \ref{Fig.3.1}, where $x_i$, $i=1,2,\cdots,9$ are players. Each player chooses his next opinion  $1$ (with whit circle) for ``agree" and $0$ (with block circle) for ``disagree" based on its neighborhood information.
The boundary players $A,B,C,D,E,F$ have invariant opinion $1$, and $U,V,W,X,Y,Z$ have invariant opinion $0$.

Each player always follows the majority. Counting himself, a player  has $5$ neighbors. So the decision is unique. Note that they might have boundary neighbors, who have fixed attitude.

\vskip 5mm

\begin{figure}
\centering
\setlength{\unitlength}{8mm}
\begin{picture}(11,11)\thicklines
\put(1,3){\circle*{0.2}}
\put(1,5){\circle*{0.2}}
\put(1,7){\circle*{0.2}}
\put(9,3){\circle*{0.2}}
\put(9,5){\circle*{0.2}}
\put(9,7){\circle*{0.2}}
\put(3,1){\circle{0.2}}
\put(5,1){\circle{0.2}}
\put(7,1){\circle{0.2}}
\put(3,9){\circle{0.2}}
\put(5,9){\circle{0.2}}
\put(7,9){\circle{0.2}}
\put(3,3){\circle{0.2}}
\put(5,3){\circle{0.2}}
\put(7,3){\circle{0.2}}
\put(3,5){\circle{0.2}}
\put(5,5){\circle{0.2}}
\put(7,5){\circle{0.2}}
\put(3,7){\circle{0.2}}
\put(5,7){\circle{0.2}}
\put(7,7){\circle{0.2}}
\put(1,3){\line(1,0){8}}
\put(1,5){\line(1,0){8}}
\put(1,7){\line(1,0){8}}
\put(3,1.1){\line(0,1){7.8}}
\put(5,1.1){\line(0,1){7.8}}
\put(7,1.1){\line(0,1){7.8}}
\put(3.2,9.2){$A$}
\put(5.2,9.2){$B$}
\put(7.2,9.2){$C$}
\put(3.2,1.2){$D$}
\put(5.2,1.2){$E$}
\put(7.2,1.2){$F$}
\put(1.2,7.2){$U$}
\put(1.2,5.2){$V$}
\put(1.4,4.5){$u(t)$}
\put(1.2,3.2){$W$}
\put(9.2,7.2){$X$}
\put(9.2,5.2){$Y$}
\put(9.2,3.2){$Z$}
\put(1,5){\vector(1,0){1.8}}
\put(3.2,7.2){$x_1$}
\put(5.2,7.2){$x_2$}
\put(7.2,7.2){$x_3$}
\put(3.2,5.2){$x_4$}
\put(5.2,5.2){$x_5$}
\put(7.2,5.2){$x_6$}
\put(3.2,3.2){$x_7$}
\put(5.2,3.2){$x_8$}
\put(7.2,3.2){$x_9$}
\end{picture}
\caption{Social Network}\label{Fig.3.1}
\end{figure}
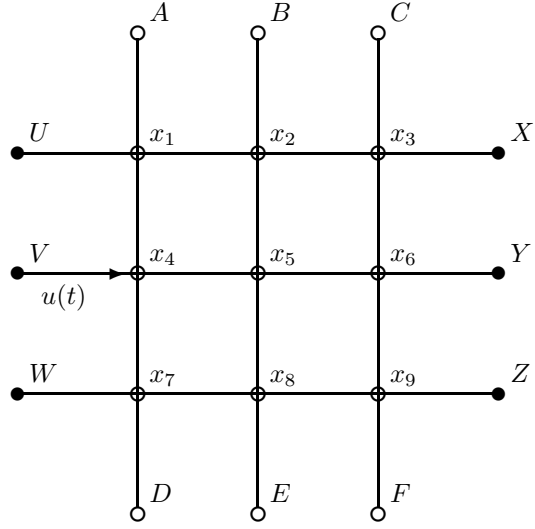

\vskip 5mm

Using ASSR, we have
\begin{align}\label{3.2.201}
x(t+1)=Mx(t),
\end{align}
where $x=\ltimes_{i=1}^9x_i$ and $M\in {\cal M}_{256\times 256}$ is in Appendix.

Now assume we are particularly interested in three situations:
$S:=\{x^1,x^2,x^3\}$, where
$$
\begin{array}{l}
x^1=\d_{512}^{43}\sim\{1,1,1,0,1,0,1,0,0\},\\
x^2=\d_{512}^{143}\sim\{1,0,1,1,0,1,1,1,1\},\\
x^3=\d_{512}^{165}\sim\{1,0,1,0,1,1,0,1,1\}.
\end{array}
$$
Then the  index function for $S$ is defined as
$$
g_1(x)=
\begin{cases}
1,\quad x\in S,\\
0,\quad \mbox{Otherwise}.
\end{cases}
$$
Correspondingly, we have its structure matrix
$$
\Col_i(G_1)=
\begin{cases}
\d_2^1,\quad \d_{512}^i\in S,\\
\d_2^2,\quad \mbox{Otherwise}.
\end{cases}
$$

Then $G_2=G_1M$ can be expressed as
$$
\Col_i(G_2)=
\begin{cases}
\d_2^1,\quad i=22,89,150,278,\\
\d_2^2,\quad \mbox{Otherwise}.
\end{cases}
$$
Furthermore,
$$
G_2M=G_1.
$$
Set $z=z_1z_2$, where
$$
z_1=G_1x,\quad z_2=G_3x,
$$
with $x=\ltimes_{i=1}^{9}x_i$.
It follows that
$$
\begin{array}{ccl}
z_1(t+1)&=&G_1x(t+1)\\
~&=&G_1Mx(t)=G_2x(t)\\
~&=&z_2(t).
\end{array}
$$
$$
\begin{array}{ccl}
z_2(t+1)&=&G_2x(t+1)\\
~&=&G_2Mx(t)=G_1x(t)\\
~&=&z_1(t).
\end{array}
$$

Hence the smallest $M$ invariant subspace containing $g_1$ is
$$
G={\cal F}_{\ell}\{g_1,g_2\}.
$$
The aggregated system becomes
\begin{align}\label{3.2.3}
\begin{array}{l}
z_1(t+1)=z_2(t)=(\J_2^T\otimes I_2)z(t)\\
z_2(t+1)=z_1(t)=(I_2\otimes \J_2^T)z(t),
\end{array}
\end{align}
where $z(t)=z_1(t)z_2(t)$.
Hence, the ASSR of $z(t)$ is
\begin{align}\label{3.2.4}
z(t+1)=\left[(\J_2^T\otimes I_2)*(I_2\otimes \J_2^T)\right]z(t)=\d_4[1,3,2,4]z(t).
\end{align}
The aggregated BN (\ref{3.2.4}) is much smaller that the original BN (\ref{3.2.2}), but it is enough to describe the dynamics of the state $z^*=g_1(x)$, which is concerned by us.

\end{exa}

\begin{rem}\label{r3.2.6} From Example \ref{e3.2.5} one sees easily that as the related attractor of a BN is of small size, then the aggregated BN might reduce the size of the original BN tremendously. Now one may ask if the related attractor is of large size, then what can we do? Of course, if the $M$ invariant subspace containing the separating logical functions, which represent the properties interesting to us, involves large size attractors, then the aggregated BN may still have large scale. Fortunately, as pointed by Kauffman \cite{kau95}:  The ``vast order" of a large scale cellular network is decided by ``tiny attractors". This fact makes the aggregation technique more useful.
\end{rem}

\section{Invariant Subspace of BCN}

Consider BCN (\ref{2.2.5}) with its ASSR (\ref{2.2.7}). Splitting $L$ into $2^m$ blocks as
\begin{align}\label{4.1}
L=[M_1,M_2,\cdots,M_{2^m}],
\end{align}
where
$$
M_r=L\d_{2^m}^r\in {\cal L}_{2^n\times 2^n},\quad r=1,2,\cdots,2^m.
$$

\begin{dfn}\label{d4.1}
\begin{itemize}
\item[(i)] ${\cal Z}$ is said to be $L$ invariant, if ${\cal Z}$ is $M_i$ invariant for all $i=1,2,\cdots,2^m$.

\item[(ii)] ${\cal Z}$ is said to be partly $L$ invariant with respect to $U\subset\d_{2^m}\{1,2,\cdots,2^m\}$, if ${\cal Z}$ is $M_i$ invariant for all $u\in U$.
\end{itemize}
\end{dfn}

\begin{dfn}\label{d4.2}
\begin{itemize}
\item[(i)] ${\cal V}$ is called a control invariant subspace containing ${\cal Z}$, if it contains ${\cal Z}$, and  for any control $u$ it is
$Lu$ invariant.
\item[(ii)] The intersection of all control invariant subspaces containing ${\cal Z}$ is called the smallest control invariant subspace containing ${\cal Z}$, and denoted by $\overline{\cal Z}$.
\item[(iii)]  ${\cal V}$ is called a partly control invariant subspace containing ${\cal Z}$ with respect to $U$, if it contains ${\cal Z}$, and  for any control $u\in U$ it is
$Lu$ invariant.
\item[(ii)] The intersection of all partly control invariant subspaces containing ${\cal Z}$ with respect to $U$ is called the smallest partly control invariant subspace containing ${\cal Z}$ with respect to $U$, and denoted by $\overline{\cal Z}^U$.

\end{itemize}
\end{dfn}

Assume ${\cal Z}={\cal F}_{\ell}\{z_1,z_2,\cdots,z_r\}$ and $\overline{\cal Z}={\cal F}_{\ell}\{z_1,\cdots,z_r,z_{r+1},\cdots,z_s\}$.
Denote  $z=\ltimes_{i=1}^sz_i$, then there exists a $G\in {\cal L}_{2^s\times 2^n}$, such that
$$
z=G\ltimes_{i=1}^nx_i:=Gx.
$$
Since $\overline{\cal Z}$ is control invariant subspaces, for $u=\d_{2^m}^i$ we have
\begin{align}\label{4.2}
GM_i=H_iG,\quad i=1,2,\cdots,2^m.
\end{align}

It follows that
$$
\begin{array}{ccl}
z(t+1)&=&Gx(t+1)=GLu(t)x(t)\\
~&=&[H_1,H_2,\cdots,H_{2^m}]Gx(t)\\
~&=&[H_1,H_2,\cdots,H_{2^m}]u(t)z(t)\\
\end{array}
$$
Define $H:=[H_1,H_2,\cdots,H_{2^m}]$, then we have the aggregated BCN as
\begin{align}\label{4.3}
z(t+1)=Hu(t)z(t).
\end{align}

Next, we consider the case when there is a constrain on control, as $u(t)\in U\subset \D_{2^m}$.
Assume $U$ is state-depending. That is,
$$
U=\left\{u\neq \d_{2^m}^\a ~\mbox{if}~ z\in X_{\a}\subset {\cal X}=\D_{2^n}\;|\;\a\in \Xi\subset \D_{2^m}\right\}.
$$

We need the following notation: $A\in {\cal M}_{p\times q}$ is called a zero-extended logical matrix if
$$
\Col(A)\subset \D_p\cup \0_p.
$$
That is $A$ may contain some zero columns.

Now consider partly control invariant subspaces containing ${\cal Z}$. Assume when $z=\d_{2^s}^k$, $u=\d_{2^m}^{\a}$ is forbidden. Then in equation (\ref{4.3}) we set
$$
\Col_k(H_{\a})=\0_{2^s}.
$$
Finally, we can construct the modified $H$, denoted by $H^U$, to describe the partly control invariant aggregated BCN, which has its dynamic equation as
\begin{align}\label{4.4}
z(t+1)=H^Uu(t)z(t).
\end{align}

We use an example to depict it.

\begin{exa}\label{e4.3} Recall Example \ref{e3.2.5}. Assume the boundary player $V$ is replaced by a control $u(t)$, ( refer to Fig. \ref{Fig.3.1}).

Then it is a normal routine to figure out the dynamics of this BCN as
\begin{align}\label{4.5}
x(t+1)=[N,M]u(t)x(t),
\end{align}
where $M$ is the same as in Example \ref{e3.2.5}, $N$ is also in Appendix.

Assume we are still particularly interested in  the $S$ as in Example \ref{e3.2.5}, i.e.,
$S:=\{x^1,x^2,x^3\}$,
where $x^1=\d_{512}^{43}$, $
x^2=\d_{512}^{143}\sim\{1,0,1,1,0,1,1,1,1\}$,
$x^3=\d_{512}^{165}$.

Then it is easy to calculate that
$G_1N=G_3$, $G_3N=G_4$, $G_4N=G5$, $G_5N=G_7$; $G_2N=G_6$, $G_6N=G_5$; $G_7M=G_7$, $G_7N=G_7$,
where
$$
\Col_i(G_3)=
\begin{cases}
\d_2^1,\quad i=43,47,143,164,229,420,\\
\d_2^2,\quad \mbox{Otherwise}.
\end{cases}
$$
$$
\Col_i(G_4)=
\begin{cases}
\d_2^1,\quad i=59,118,278,\\
\d_2^2,\quad \mbox{Otherwise}.
\end{cases}
$$
$$
\Col_i(G_5)=
\begin{cases}
\d_2^1,\quad i=164,299,420,\\
\d_2^2,\quad \mbox{Otherwise}.
\end{cases}
$$
$$
\Col_i(G_6)=
\begin{cases}
\d_2^1,\quad i=278,\\
\d_2^2,\quad \mbox{Otherwise}.
\end{cases}
$$

$$
\Col_i(G_7)=\d_2^2,\quad i=1,2,\cdots,512.
$$

Set
$$
z_i=G_ix,\quad i=1,2,\cdots,7,
$$
$W=\{3,4,5,6\}$
and we assume the feasible control set
$$
U=\{u(t)\neq \d_2^2\;|\;x(t)\in W\}
$$

Finally, the partly control invariant aggregation BCN, is obtained as follows.
\begin{align}\label{4.6}
z(t+1)=H^Uu(t)z(t),
\end{align}
where
$z(t)=(z_1(t),z_2(t),z_3(t),z_4(t),z_5(t),z_6(t),z_7(t))^T$, and
$$
H^U=\d_7[6,3,4,5,7,5,7,2,1,0,0,0,0,7].
$$

The state-transition graph is depicted in Fig. \ref{Fig.4.1}.

\vskip 2mm

\begin{figure}
\centering
\setlength{\unitlength}{8 mm}
\begin{picture}(16,10)\thicklines
\put(4,2.5){\oval(2,1)}
\put(4,5.5){\oval(2,1)}
\put(4,8.5){\oval(2,1)}
\put(9,2.5){\oval(2,1)}
\put(9,5.5){\oval(2,1)}
\put(9,8.5){\oval(2,1)}
\put(14,5.5){\oval(2,1)}
\put(4,8){\vector(0,-1){2}}
\put(4,5){\vector(0,-1){2}}
\put(9,8){\vector(0,-1){2}}
\put(9,3){\vector(0,1){2}}
\put(5,2.5){\vector(1,0){3}}
\put(5,8.5){\vector(1,0){3}}
\put(10,5.5){\vector(1,0){3}}
\put(1,5.5){\line(1,0){2}}
\put(1,8.5){\vector(1,0){2}}
\put(1,5.5){\line(0,1){3}}
\put(3,2.5){\vector(-1,0){2}}
\put(1.75,2.75){\line(1,-1){0.5}}
\put(1.75,2.25){\line(1,1){0.5}}
\put(8,5.5){\vector(-1,0){2}}
\put(6.75,5.75){\line(1,-1){0.5}}
\put(6.75,5.25){\line(1,1){0.5}}
\put(10,8.5){\vector(1,0){2}}
\put(10.75,8.75){\line(1,-1){0.5}}
\put(10.75,8.25){\line(1,1){0.5}}
\put(10,2.5){\vector(1,0){2}}
\put(10.75,2.75){\line(1,-1){0.5}}
\put(10.75,2.25){\line(1,1){0.5}}
\put(14,4.5){\oval(1,1)[b]}
\put(14.5,4.5){\line(0,1){0.5}}
\put(13.5,4.5){\vector(0,1){0.5}}
\put(3.8,8.5){$z_1$}
\put(3.8,5.5){$z_2$}
\put(3.8,2.5){$z_3$}
\put(8.8,8.5){$z_6$}
\put(8.8,5.5){$z_5$}
\put(8.8,2.5){$z_4$}
\put(13.8,5.5){$z_7$}
\put(6,9){$u=1$}
\put(6,3){$u=1$}
\put(10.5,9){$u=0$}
\put(10.5,6){$u=1$}
\put(10.5,3){$u=0$}
\put(13.5,3.5){$u=1$}
\put(13.5,3){$u=0$}
\put(1.2,7){$u=0$}
\put(4.2,7){$u=0$}
\put(1.5,1.5){$u=0$}
\put(4.2,4){$u=1$}
\put(10.5,1.5){$u=0$}
\put(6.5,6){$u=0$}
\end{picture}
\caption{State Transition Graph of aggregated BCN (\ref{4.6}) )\label{Fig.4.1}}
\end{figure}
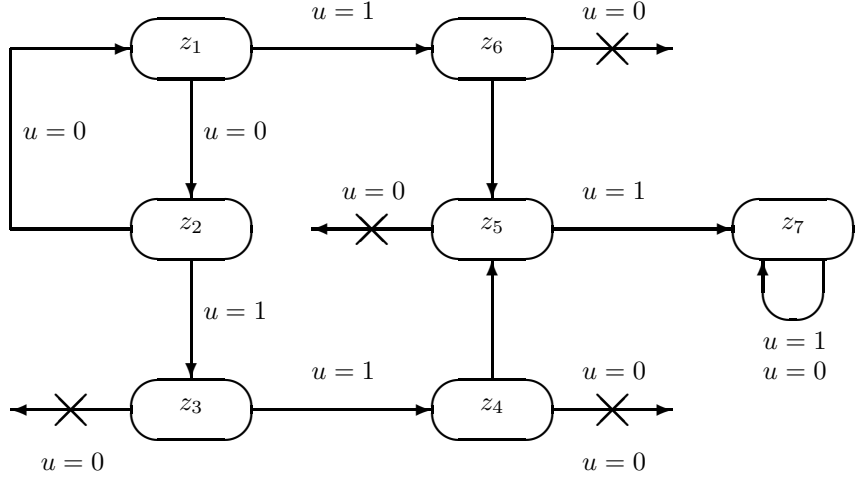

\vskip 2mm

\end{exa}

\section{Minimum Realization of BCN}

Consider a BNC (\ref{2.2.5}) with outputs (observers)
\begin{align}\label{5.1}
\begin{array}{l}
y_1(t)=\xi_1(x_1(t),\cdots,x_n(t)),\\
y_2(t)=\xi_2(x_1(t),\cdots,x_n(t)),\\
\vdots\\
y_r(t)=\xi_r(x_1(t),\cdots,x_n(t)).\\
\end{array}
\end{align}

Then the input-output BCN (\ref{2.2.5})-(\ref{5.1}) has  ASSR as
\begin{align}\label{5.2}
\begin{cases}
x(t+1)=Lu(t)x(t),\\
y(t)=Hx(t).
\end{cases}
\end{align}
After a coordinate change $T:x\ra z$, expressed by $z=Tx$, where $T\in {\cal L}_{2^n\times 2^n}$, (\ref{2.2.7}) becomes \cite{che11}
\begin{align}\label{5.3}
\begin{cases}
z(t+1)=\tilde{L}u(t)z(t),\\
y(t)=\tilde{H}z(t),
\end{cases}
\end{align}
where
$$
\begin{array}{l}
\tilde{L}=TL\left(I_{2^m}\otimes T^T\right),\\
\tilde{H}=HT^T.
\end{array}
$$

If under the coordinate frame $z$ (\ref{5.1}), expressed as (\ref{5.3}),  has the form of (\ref{3.1.2}), then it is clear that ${\cal Z}={\cal F}_{\ell}\{z^1\}$ is a control invariant subspace, containing ${\cal Y}$. In fact, we can ignore $z^2$ and give the following definition.

\begin{dfn}\label{d5.1} Consider BCN (\ref{5.1}), if there exists a subspace ${\cal Z}={\cal F}_{\ell}\{z^1_1,x^1_2,\cdots,z^1_r\}$ such that
\begin{align}\label{5.4}
\begin{cases}
z^1(t+1)=F^1(z^1,u),\\
y(t)=\xi(z^1(t)),
\end{cases}
\end{align}
then (\ref{5.4}) is called a realization of (\ref{2.2.5})-(\ref{5.1}).
\end{dfn}

\begin{rem}\label{r5.2}
\begin{itemize}
\item[(i)] From Definition \ref{d5.1}, ${\cal Z}$ is a control invariant subspace, containing ${\cal Y}$.

\item[(ii)] In Definition \ref{d5.1} ${\cal Z}$ is not required to be a regular subspace.

\item[(iii)] It is obvious that (\ref{5.4}) and (\ref{2.2.5})-(\ref{5.1}) have the same input-output mapping.
\end{itemize}
\end{rem}

\begin{dfn}\label{d5.2} Consider BCN (\ref{5.1}), if  ${\cal Z}={\cal F}_{\ell}\{z^1_1,x^1_2,\cdots,z^1_r\}$ is the smallest control invariant subspace containing ${\cal Y}$, then the corresponding BN (\ref{5.4}) is called the minimum realization of (\ref{2.2.5})-(\ref{5.1}).
\end{dfn}

\begin{prp}\label{p5.3} Assume ${\cal Z}={\cal F}_{\ell}\{z^1_1,x^1_2,\cdots,z^1_r\}$ is the smallest control invariant subspace containing ${\cal Y}$ and ${\cal Z}=Gx$. then
\begin{itemize}
\item[(i)]
there exists a set of logical matrix $H_i\in {\cal L}_{r\times r}$, $\i=1,2,\cdots,2^m$ such that
\begin{align}\label{5.401}
GM_i=H_iG,\quad i=1,2,\cdots,2^m;
\end{align}
\item[(ii)] the minimum realization of (\ref{2.2.5})-(\ref{5.1}) has its ASSR  as
\begin{align}\label{5.5}
\begin{cases}
z^1(t+1)=Hu(t)z^1(t),\\
y(t)=\Xi z^1(t),
\end{cases}
\end{align}
where $\Xi$ is the structure matrix of $\xi$, and
$$
H=[H_1,H_2,\cdots,H_{2^m}].
$$
\end{itemize}
\end{prp}

The following algorithm provides a way to construct the minimum realization of a BCN.

\begin{alg}\label{a5.4}
\begin{itemize}
\item Step 1:

Set
$$
{\cal O}_0=\left\{y_1,y_2,\cdots,y_p\right\}.
$$
Calculate
$$
{\cal O}_1=\{yM_1,yM_2,\cdots,YM_{2^m}\;|\;y\in {\cal O}_0\}\backslash \{{\cal O}_0\}.
$$
\item Step s: ($s>0$)

Calculate
$$
{\cal O}_{s+1}=\{yM_1,yM_2,\cdots,yM_{2^m}\;|\;y\in {\cal O}_s\}\backslash \{ {\cal O}_r\;|\;r=,0,1,\cdots,s\}.
$$

\item Last Step.
If
$$
{\cal O}_{s^*+1}=\emptyset.
$$
then
$$
{\cal Z}^*:={\cal F}_{\ell}\{{\cal O}_{r}\;|\; r=0,1,\cdots,s^*\}
$$
is the smallest control invariant subspace containing $y$.
\end{itemize}
\end{alg}

Assume ${\cal Z}^*={\cal F}_{\ell}\{z_1,z_2,\cdots,z_r\}$, set
$z=\ltimes_{i=1}^rz_i$, then 
\begin{align}\label{5.6}
\begin{array}{l}
z(t+1)=\left[H_1,H_2,\cdots,H_{2^m}\right]u(t)z(t),\\
y(t)=\Xi z(t)
\end{array}
\end{align}
is the minimum realization of BCN (\ref{2.2.5})-(\ref{5.1}).

Next, we consider an example.

\begin{exa}\label{e5.5}

Consider a BCN, with its ASSR as
\begin{align}\label{5.7}
\begin{cases}
x(t+1)=Lu(t)x(t),\\
y(t)=\Xi x(t),
\end{cases}
\end{align}
where
$x(t)=\ltimes_{i=1}^nx_i(t)$, $u(t)=u_1(t)u_2(t)$, and
$$
L=[M_1,M_2,M_3,M_4],
$$
with
$$
\begin{array}{cc}
M_1=
\left(\begin{array}{cc}
\begin{bmatrix}
0&0&1\\
1&0&0\\
0&1&0\\
\end{bmatrix}&\0\\
\0&{\bf X}\\
\end{array}\right),
&
M_2=
\left(\begin{array}{cc}
\begin{bmatrix}
0&1&0\\
1&0&0\\
0&0&1\\
\end{bmatrix}&\0\\
\0&{\bf X}\\
\end{array}\right),\\
M_3=
\left(\begin{array}{cc}
\begin{bmatrix}
1&0&0\\
0&1&0\\
0&0&1\\
\end{bmatrix}&\0\\
\0&{\bf X}\\
\end{array}\right),
&
M_4=
\left(\begin{array}{cc}
\begin{bmatrix}
0&0&1\\
0&1&0\\
1&0&0\\
\end{bmatrix}&\0\\
\0&{\bf X}\\
\end{array}\right),\\
\end{array}
$$
where ${\bf X}\in {\cal L}_{(2^n-3)\times (2^n-3)}$ is an uncertain logical matrix.
$$
\Xi=\d_2[1,2,1,\underbrace{2,2,\cdots,2}_{2^n-3}].
$$

Denote $y_1=y$, then it is easy to calculate that
$$
\begin{array}{l}
y_1M_1=\d_2[2,1,1,2,\cdots,2]:=y_2,\\
y_1M_2=y_2,\\
y_1M_3=y_1,\\
y_1M_4=y_1,\\
y_2M_1=\d_2[1,1,2,2,\cdots,2]:=y_3,\\
y_2M_2=y_1,\\
y_2M_3==y_2,\\
y_2M_4==y_3,\\
y_3M_1=y_1,\\
y_3M_2==y_3,\\
y_3M_3=y_3,\\
y_3M_4=y_2.\\
\end{array}
$$
Let
$$
z_1=y_1,\quad z_2=y_2,\quad z_3=y_3.
$$
Hence we have
$$
\begin{array}{l}
z_1(t+1)=[Z_2,Z_2,Z_1,Z_1]u(t)z(t),\\
z_2(t+1)=[Z_3,Z_1,Z_2,Z_3]u(t)z(t),\\
z_3(t+1)=[Z_1,Z_3,Z_3,Z_2]u(t)z(t),\\
\end{array}
$$
where $u(t)=u_1(t)u_2(t)$, $z(t)=z_1(t)z_2(t)z_3(t)$, and
$$
\begin{array}{l}
Z_1=I_2\otimes \J^T_4=\d_2[1,1,1,1,2,2,2,2],\\
Z_2=\J_2\otimes I_2\otimes \J_2=\d_2[1,1,2,2,1,1,2,2],\\
Z_3=\J_4\otimes I_2=\d_2[1,2,1,2,1,2,1,2].
\end{array}
$$

Finally, the minimum realization of (\ref{5.7}) is obtained as
\begin{align}\label{5.8}
\begin{cases}
z(t+1)=L^*u(t)z(t),\\
y(t)=Z_1z(t),
\end{cases}
\end{align}
where
$$
\begin{array}{ccl}
L^*&=\d_8[&1,3,5,7,2,4,6,8,1,2,5,6,3,4,7,8,\\
~&~       &1,2,3,4,5,6,7,8,1,3,2,4,5,7,6,8].
\end{array}
$$

The state-transition graph of this minimum realization is depicted in Fig. \ref{Fig.5.1}.

\vskip 2mm

\begin{figure}
\centering
\setlength{\unitlength}{8 mm}
\begin{picture}(9,12)\thicklines
\put(2,2.5){\oval(2,1)}
\put(2,8.5){\oval(2,1)}
\put(7,8.5){\oval(2,1)}
\put(1.5,3){\vector(0,1){5}}
\put(6.5,8){\vector(-1,-1){5}}
\put(2.5,3){\vector(1,1){5}}
\put(6,8.25){\vector(-1,0){3}}
\put(3,8.75){\vector(1,0){3}}
\put(1.5,3){\vector(0,1){5}}
\put(2,1.5){\oval(1,1)[b]}
\put(2.5,1.5){\line(0,1){0.5}}
\put(1.5,1.5){\vector(0,1){0.5}}
\put(2,9.5){\oval(1,1)[t]}
\put(1.5,9.5){\line(0,-1){0.5}}
\put(2.5,9.5){\vector(0,-1){0.5}}
\put(7,9.5){\oval(1,1)[t]}
\put(6.5,9.5){\line(0,-1){0.5}}
\put(7.5,9.5){\vector(0,-1){0.5}}
\put(1,8.5){\vector(-1,0){1}}
\put(0.2,8.8){$y$}
\put(1.8,8.5){$z_1$}
\put(1.8,2.5){$z_2$}
\put(6.8,8.5){$z_3$}
\put(3.5,9.2){$u=\d_4^1,\d_4^2$}
\put(4,7.8){$u=\d_4^2$}
\put(0,5.5){$u=\d_4^1$}
\put(2.5,6){$u=\d_4^1,\d_4^2$}
\put(4.5,4.5){$u=\d_4^4$}
\put(1,10.2){$u=\d_4^3,\d_4^4$}
\put(1,0.5){$u=\d_4^2,\d_4^3$}
\put(6.5,10.2){$u=\d_4^3$}
\end{picture}
\caption{State Transition Graph of aggregated BCN (\ref{5.8}) )\label{Fig.5.1}}
\end{figure}
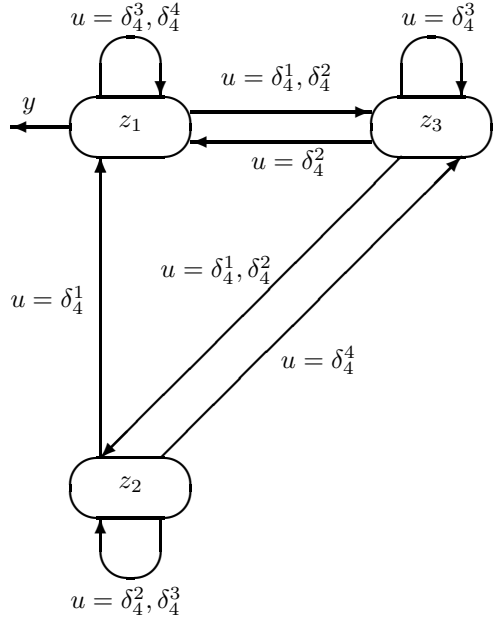

\end{exa}

Motivated by Example \ref{e5.5}, the following result is easily verifiable.

\begin{prp}\label{p5.6} Consider BCN (\ref{2.2.5})-(\ref{5.1}). If there exists a coordinate change
$$
z=Tx,
$$
such that
$$
TM_iT^T=\begin{bmatrix}
J_1^i&\0&\cdots&\0\\
\0&J_2^i&\cdots&\0\\
~&~&\ddots&\0\\
\0&\0&\cdots&J_s^i
\end{bmatrix},\quad i=1,2,\cdots,2^m,
$$
where $z=(z^1,z^2,\cdots,z^s)$ and $z^k$ corresponds to  $J_k^i$. Moreover, if $y\in {\cal F}_{\ell}\{z^k\}$, then there exists a realization
\begin{align}\label{5.9}
\begin{cases}
z(t+1)=\left[J^1_k,J^2_k,\cdots,J^{2^m}_k\right]u(t)z(t),\\
y(t)=\Xi_kz(t).
\end{cases}
\end{align}
Moreover, if $J^i_k$, $i=1,2,\cdots,2^m$ can not be further diagonized simultaneously for any $1\leq k\leq s$, then (\ref{5.9}) is a minimum realization.
\end{prp}

\begin{rem}\label{r5.7}
\begin{itemize}
\item[(i)] Minimum realization can also be considered as a kind of aggregations, which separate states into two categories: related states and unrelated states. Then only related states are modeled.

\item[(ii)] For a large scale BN, we can inject controls on different nodes and observe some other nodes (which are considered as outputs). Then observe the input-output relations to investigate the minimum realization, which reveals part structure of the BN. By changing input nodes and output nodes, another part structure may be revealed. The minimum realizations might be of much smaller sizes, which makes the investigations easier. This method may provide a way to solve the problem of computational complexity.

\end{itemize}
\end{rem}

An alternative way to deal with a large-scale BN is to observe some special interested states. Then the observers, as a set of logical functions, can be considered as separating functions. Then we may define the follows:

\begin{dfn}\label{d5.8} A BN with some outputs is called an observe-based BN. Using observers as a set of logical functions, the dynamic equations of the minimum-invariant subspace containing observers are called the observe-based minimum realization of the observe-based BN.
\end{dfn}

In fact, through observed data, we may construct the dynamic equations of the observe-based minimum realization. In this way, part of the structure of overall BN can be constructed. Using different observes, the interested parts of structure of overall BN might be construct.

\section{Conclusion}

In this paper a logical function is considered as an index function of a subset of nodes of a BN or BCN. Using this idea, a set of logical functions are used as separating functions to aggregate nodes. Then the (minimum) invariant subspace containing the preassigned set of logical functions, is constructed. Furthermore, the dynamic equations for the invariant subspace, which represents the aggregated nodes, are obtained. Then the (minimum) invariant subspace of BNC is also defined and the corresponding dynamic equations are also constructed. Finally, as the outputs of a BC/BNC are considered as the set of separating functions, the minimum realization of a BCN (or the observe-based minimum realization for BC) is defined, and their properties are investigated.

When a BN/BCN is of large scale, the structure matrix of overall BN might be huge and practically uncomputable. Using input-output realization and observe-based minimum realization, the interested parts of structure of the BN could be obtained. These might be much smaller sub-BN may dominate the behaviors of whole BN. Hence, this technique may provide an efficient way to solve the computational complexity of large scale BN/BCN. 

\vskip 5mm

{\bf Acknowledgment} This work was completed when the second and third authors visiting the Center of STP Theory and Applications.

\section{Appendix}

\begin{itemize}
\item[(i)] The structure matrix of BN (\ref{3.2.2}):
\begin{tiny}
$$
\begin{array}{lrrrrrrrrrrrrrrrr}
M=\d_{512}[&  1&  1&  1&  2&  1&  1&  5&  8&  1& 10&  2& 10&  1& 10&  6& 16\\
          ~&  1&  9&  1& 12& 33& 43& 39& 48&  9& 10& 26& 28& 41& 44& 64& 64\\
          ~&  1&  1&  5&  6& 37& 37& 37& 40&  1& 10& 22& 30& 37& 46& 54& 64\\
          ~& 33& 41& 53& 64& 37& 47& 55& 64& 57& 58& 62& 64& 61& 64& 64& 64\\
          ~&  1&  9&  1& 10&  1&  9&  5& 16& 73& 74& 74& 74& 73& 74& 78& 80\\
          ~&  9&  9&  9& 12& 41& 43& 47& 48& 73& 74& 90& 92&105&108&128&128\\
          ~&  1&  9&  5& 14& 37& 45& 37& 48& 73& 74& 94& 94&109&110&126&128\\
          ~& 41& 41& 61& 64& 45& 47& 63& 64&121&122&126&128&125&128&128&128\\
          ~&  1&  1&  1&  2&  1&  1&  5&  8& 65& 74& 82& 90& 65& 74& 86& 96\\
          ~&  1&  9& 17& 28& 33& 43& 55& 64& 89& 90& 90& 92&121&124&128&128\\
          ~&257&257&277&278&293&293&309&312&337&346&342&350&373&382&374&384\\
          ~&305&313&309&320&309&319&311&320&377&378&382&384&381&384&384&384\\
          ~& 65& 73& 65& 74& 65& 73& 69& 80& 73& 74& 90& 90& 73& 74& 94& 96\\
          ~&201&201&217&220&233&235&255&256&217&218&218&220&249&252&256&256\\
          ~&321&329&341&350&357&365&373&384&345&346&350&350&381&382&382&384\\
          ~&505&505&509&512&509&511&511&512&505&506&510&512&509&512&512&512\\
          ~&  1&  1&  1&  2& 33& 33& 37& 40&  1& 10&  2& 10& 33& 42& 38& 48\\
          ~& 33& 41& 33& 44& 33& 43& 39& 48& 41& 42& 58& 60& 41& 44& 64& 64\\
          ~&289&289&293&294&293&293&293&296&289&298&310&318&293&302&310&320\\
          ~&289&297&309&320&293&303&311&320&313&314&318&320&317&320&320&320\\
          ~&  1&  9&  1& 10& 33& 41& 37& 48& 73& 74& 74& 74&105&106&110&112\\
          ~&169&169&169&172&169&171&175&176&233&234&250&252&233&236&256&256\\
          ~&289&297&293&302&293&301&293&304&361&362&382&382&365&366&382&384\\
          ~&425&425&445&448&429&431&447&448&505&506&510&512&509&512&512&512\\
          ~&257&257&257&258&289&289&293&296&321&330&338&346&353&362&374&384\\
          ~&417&425&433&444&417&427&439&448&505&506&506&508&505&508&512&512\\
          ~&289&289&309&310&293&293&309&312&369&378&374&382&373&382&374&384\\
          ~&433&441&437&448&437&447&439&448&505&506&510&512&509&512&512&512\\
          ~&449&457&449&458&481&489&485&496&457&458&474&474&489&490&510&512\\
          ~&489&489&505&508&489&491&511&512&505&506&506&508&505&508&512&512\\
          ~&481&489&501&510&485&493&501&512&505&506&510&510&509&510&510&512\\
          ~&505&505&509&512&509&511&511&512&505&506&510&512&509&512&512&512].
\end{array}
$$
\end{tiny}

\item[(ii)] Structure matrix for Example \ref{e4.3}.

\begin{tiny}
$$
\begin{array}{lrrrrrrrrrrrrrrrr}
N=\d_{512}[&  1&  1&  1&  2&  1&  1&  5&  8&  1& 10&  2& 10&  1& 10&  6& 16\\
            ~&  1&  9&  1& 12&  1& 11&  7& 16&  9& 10& 26& 28&  9& 12& 32& 32\\
            ~&  1&  1&  5&  6&  5&  5&  5&  8&  1& 10& 22& 30&  5& 14& 22& 32\\
            ~&  1&  9& 21& 32& 37& 47& 55& 64& 25& 26& 30& 32& 61& 64& 64& 64\\
            ~&  1&  9&  1& 10&  1&  9&  5& 16& 73& 74& 74& 74& 73& 74& 78& 80\\
            ~&  9&  9&  9& 12&  9& 11& 15& 16& 73& 74& 90& 92& 73& 76& 96& 96\\
            ~&  1&  9&  5& 14&  5& 13&  5& 16& 73& 74& 94& 94& 77& 78& 94& 96\\
            ~&  9&  9& 29& 32& 45& 47& 63& 64& 89& 90& 94& 96&125&128&128&128\\
            ~&  1&  1&  1&  2&  1&  1&  5&  8& 65& 74& 82& 90& 65& 74& 86& 96\\
            ~&  1&  9& 17& 28&  1& 11& 23& 32& 89& 90& 90& 92& 89& 92& 96& 96\\
            ~&257&257&277&278&261&261&277&280&337&346&342&350&341&350&342&352\\
            ~&273&281&277&288&309&319&311&320&345&346&350&352&381&384&384&384\\
            ~& 65& 73& 65& 74& 65& 73& 69& 80& 73& 74& 90& 90& 73& 74& 94& 96\\
            ~&201&201&217&220&201&203&223&224&217&218&218&220&217&220&224&224\\
            ~&321&329&341&350&325&333&341&352&345&346&350&350&349&350&350&352\\
            ~&473&473&477&480&509&511&511&512&473&474&478&480&509&512&512&512\\
            ~&  1&  1&  1&  2&  1&  1&  5&  8&  1& 10&  2& 10&  1& 10&  6& 16\\
            ~&  1&  9&  1& 12& 33& 43& 39& 48&  9& 10& 26& 28& 41& 44& 64& 64\\
            ~&257&257&261&262&293&293&293&296&257&266&278&286&293&302&310&320\\
            ~&289&297&309&320&293&303&311&320&313&314&318&320&317&320&320&320\\
            ~&  1&  9&  1& 10&  1&  9&  5& 16& 73& 74& 74& 74& 73& 74& 78& 80\\
            ~&137&137&137&140&169&171&175&176&201&202&218&220&233&236&256&256\\
            ~&257&265&261&270&293&301&293&304&329&330&350&350&365&366&382&384\\
            ~&425&425&445&448&429&431&447&448&505&506&510&512&509&512&512&512\\
            ~&257&257&257&258&257&257&261&264&321&330&338&346&321&330&342&352\\
            ~&385&393&401&412&417&427&439&448&473&474&474&476&505&508&512&512\\
            ~&257&257&277&278&293&293&309&312&337&346&342&350&373&382&374&384\\
            ~&433&441&437&448&437&447&439&448&505&506&510&512&509&512&512&512\\
            ~&449&457&449&458&449&457&453&464&457&458&474&474&457&458&478&480\\
            ~&457&457&473&476&489&491&511&512&473&474&474&476&505&508&512&512\\
            ~&449&457&469&478&485&493&501&512&473&474&478&478&509&510&510&512\\
            ~&505&505&509&512&509&511&511&512&505&506&510&512&509&512&512&512].
            \end{array}
            $$

\end{tiny}

\end{itemize}

\end{document}